\documentclass[runningheads]{llncs}
\usepackage[T1]{fontenc}
\usepackage{graphicx}
\usepackage{amsmath,amssymb,amsfonts}
\usepackage{hyperref}
\usepackage{url}
\usepackage{color,soul}

\usepackage{lipsum} 
\usepackage{color}

\begin{document}

\title{Watts-Per-Intelligence: Part I (Energy Efficiency)}
\titlerunning{Watts-Per-Intelligence}
\author{Elija Perrier\inst{1}\orcidID{0000-0002-6052-6798}}
\authorrunning{E. Perrier}
\institute{Centre for Quantum Software \& Information, University of Technology Sydney, Australia\\
\email{elija.perrier@gmail.com}}
\maketitle
\begin{abstract}
We present a mathematical framework for quantifying energy efficiency in intelligent systems by linking energy consumption to information‐processing capacity. We introduce a \textit{watts-per-intelligence} metric that integrates algorithmic thermodynamic principles of Landauer with computational models of machine intelligence. By formalising the irreversible energy costs of computation, we derive rigorous lower bounds on energy usage of algorithmic intelligent systems and their adaptability. We introduce theorems that constrain the trade-offs between intelligence output and energy expenditure. Our results contribute to design principles for energy‐efficient intelligent systems.
\end{abstract}

\section{Introduction}
Recent breakthroughs in artificial intelligence (AI) capabilities have been accompanied by a significant increase in the energy consumption required to achieve scalable state-of-the-art performance. As a result, improving computational energy efficiency \cite{sutton2019bitter} is now a major priority for attempts to engineer advanced AGI systems \cite{takahashi2023scenarios,chen2025much,fernandez2025energy,latif2024empirical,kurshan2023pursuit}. Yet biological structures (e.g. the human brain) use orders of magnitude less power for comparable levels of intelligence than the most capable AI models \cite{gebicke2023computational}. This difference motivates the question of how can we formally define and quantify the energy efficiency of algorithmic intelligent systems.
\subsubsection{Related work}
The relationship between energy and intelligence has inspired a broad corpus of research from biology to AI that grapple with resource constraints \cite{dyson1999origins,legg_universal_2007,hutter2004universal,hutter2024introduction,fields2025life}. Seminal work by Landauer \cite{landauer1961irreversibility} established a physical minimum for erasing one bit, which was later extended \cite{bennett1982thermodynamics} to reversible and stochastic computing. More recent work \cite{gacs1994boltzmann,levin1984randomness} laid algorithmic foundations bridging Kolmogorov complexity and thermodynamic entropy, culminating in ensemble-free algorithmic thermodynamics proposals \cite{ebtekar2024modeling,ebtekar2025foundations}. Neuromorphic engineering, thermodynamic computing, and high-performance AI research \cite{zhirnov2014minimum,wolpert2019stochastic,merolla2014million,davies2018loihi} has also underscored the practical connection with power consumption, as highlighted by large-scale model training \cite{strubell2020energy} and sample efficiency. In cognitive science, the free energy principle \cite{friston2006free,friston2019free} and related active-inference frameworks \cite{buckley2017free} tie biological intelligence to the minimisation of thermodynamic cost. However, specific work addressing (and providing bounds upon) algorithmic intelligence, power (energy) and adaptivity in terms of irreversible operations and algorithmic thermodynamic principles remains ongoing. Prior work has extended algorithmic information theory with physical measures and explored energy bounds for universal induction \cite{ozkural2021measures,ozkural2016ultimate,ozkural2015ultimate}.  Our work is complementary to but distinct from existing proposals via its focus on how the algorithmically irreversible operations inherent in intelligent activities, such as counterfactual reasoning and planning, give rise to energetic bounds on intelligence and adaptivity.

\subsubsection{Contributions}
We make the following contributions: 
\begin{enumerate}
    \item \textit{Watts-per-intelligence metric}. We introduce a novel metric, \textit{watts-per-intelligence} (WPI) ($\Phi$) quantifying energy efficiency in intelligent systems. WPI unifies Landauer's thermodynamic cost model with the formalisation of intelligence as adaptation with limited resources, enabling the energy efficiency of computational substrates producing intelligence to be compared.
    \item \textit{Algorithmic \& Adaptive Efficiency Bounds}. We establish bounds on achievable ratios of intelligence to power across both static computations and dynamically adaptive architectures.
\end{enumerate}
The remainder of the paper is organised as follows. Section \ref{sec:meas-intel} introduces a general measure of intelligence based on task performance. Section \ref{sec:wpi} defines our watts-per-intelligence metric and establishes lower bounds via overhead factors and irreversibility. Section \ref{sec:algoefficiency} presents bounds upon algorithmic efficiency in terms of entropy. Section \ref{sec:conclusion} summarises our main results and outlines future research directions toward energy-aware AGI systems.

\section{Measuring Intelligence}
\label{sec:meas-intel}
Though there is no universal consensus on the definition of intelligence, several influential formalisms exist. To illustrate our results, we adopt Legg and Hutter's universal measure $I_{\mathrm{LH}}(\pi)$ \cite{legg_universal_2007} of an intelligent agent $\pi$ which aggregates an agent's performance over all computable environments, weighted by environmental complexity. Formally:
\begin{align}
I_{\mathrm{LH}}(\pi) \;=\; \sum_{\mu \in \mathcal{M}} 2^{-K(\mu)} \, V_\mu(\pi) \label{eqn:hutterlegg}
\end{align}
where $\mathcal{M}$ is a class of all environments, $K(\mu)$ is the Kolmogorov complexity of environment $\mu$, and $V_\mu(\pi)$ denotes the performance (e.g., total reward) of agent $\pi$ in $\mu$. $I_{\mathrm{LH}}(\pi)$ is practically uncomputable. Adopting instead a \emph{finite-task} approach, let $\mathcal{T} = \{T_1, T_2, \dots, T_n\}$ be a finite set of tasks, each with a difficulty weight $w_i \ge 0$ and performance function $P_i(\pi) \in [0,1]$ (where 0 is failure and 1 is task satisfaction). To generalise our analysis, we abstract Eq. (\ref{eqn:hutterlegg}) as:
\begin{align}
I(\pi) \;=\; \sum_{i=1}^n w_i P_i(\pi) \label{eq:finiteInt}
\end{align}
allowing comparison of diverse intelligent systems based on task performance. We interpret an \emph{agent} $\pi$ as an algorithm $\mathcal{A}$ implemented on physical computational substrate architecture (e.g. hardware) $\mathcal{H}$. The overall structure $\Sigma = (\mathcal{A}, \mathcal{H})$ yields a measurable intelligence $I(\Sigma)$ that depends on both $\mathcal{A}$ and $\mathcal{H}$. We leave details of $\mathcal{H}$ general. Intuitively, it may be thought of as different hardware (or biological substrate) configurations implementing $\mathcal{A}$ intelligent behaviour $I(\Sigma)$. Or alternatively, it may be abstractly reflected in time or space complexity measures of hardware. The measure of intelligence arising from an algorithm $\mathcal{A}$ is affected by $\mathcal{H}$ e.g. a system confined to CPUs rather than GPUs will run certain $\mathcal{A}$ less efficiently (and more slowly) thus for a given time interval $\tau$ exhibit lower performance (e.g. completing less tasks), thus leading to a lower $I(\Sigma)$. For example, a system $\Sigma$ with less efficient $\mathcal{H}$ requires more energy to achieve the same $I(\mathcal{A})$.

\section{Watts per Intelligence}
\label{sec:wpi}
\subsection{Energy and irreversibility} 
Computation obeys thermodynamic laws \cite{ebtekar2025foundations}. Landauer's principle \cite{landauer1961irreversibility} states that erasing one bit of information dissipates at least $k_B T \ln 2$ joules. Fixing $\mathcal{H}$, if an algorithm $\mathcal{A}$ performs $N$ irreversible bit operations per unit time at temperature $T$, its minimal energy dissipation is proportional to $N k_B T \ln 2$. The total energy consumed performing its computation over the time interval $[0,\tau]$ is:
\begin{equation}\label{eq:landauerBound}
E(\mathcal{A}) \;\ge\; cN(\mathcal{A})
\end{equation}
where $c = k_B T \ln 2$. The power consumption over $[0,\tau]$ is:
\begin{equation}\label{eq:powerdef}
P(\mathcal{A}) \;=\; \frac{E(\mathcal{A})}{\tau}.
\end{equation}
In practice, there is additional dissipation beyond the ideal lower bound, which we reflect in the overhead factor $F(\mathcal{A})\gg 1$:
\begin{equation}\label{eq:overhead}
E(\mathcal{A}) \;=\; F(\mathcal{A}) N(\mathcal{A}) \, c
\end{equation}
where $N(\mathcal{A})$ is the minimal number of irreversible operations.
\subsection{Intelligence and irreversibility} We further assume that $I(\Sigma)$ is proportional to the number of irreversible operations implemented by an algorithm $\mathcal{A}$ on hardware $\mathcal{H}$:
\begin{align}
    I(\Sigma) \le \alpha\,N(\Sigma) \label{eqn:intellirreversible}
\end{align}
for some constant $\alpha>0$. Here $\alpha$ is akin to an algorithmic yield reflecting the useful intelligence arising from an irreversible operation. The higher $\alpha$, the higher the gain in $I(\Sigma)$ for a given irreversible operation. It also captures how well the organisation of computation converts irreversible micro-steps into high-level problem solving. The rationale for linking irreversible computations to intelligence lies in the counterfactual planning and reasoning, which we assume as necessary conditions of the models of intelligence with which we are concerned. As argued in \cite{ebtekar2024modeling}, physically, the counterfactual decision-theoretic process breaks time-reversal symmetry because merging multiple possible histories or overwriting memory states is an inherently thermodynamically directional and irreversible operation. While an agent's fine-grained dynamics may remain reversible, algorithmic equivalents of synthesising future trajectories or forgetting are irreversible as is erasing or merging memory states. In each case, they are operations that cannot be undone without additional resources. This establishes an effective net increase in algorithmic complexity (and thus energy cost) of intelligent operations - a difficulty weight - when viewed at the coarse-grained agent level.

\subsection{Watts per Intelligence}
The foregoing motivates the definition of a \textit{watts-per-intelligence} ratio $\Phi$.
\begin{definition}[Watts per intelligence]
\label{def:watts-intel}
Given the intelligence measure $I(\Sigma)$ and the power usage $P(\Sigma)$ we define the watts-per-intelligence metric  $\Phi(\Sigma)$ as:
\begin{equation}
\Phi(\Sigma) \;=\; \frac{P(\Sigma)}{I(\Sigma)}
\;=\; \frac{E(\Sigma)/\tau}{I(\Sigma)}
\;=\; \frac{E(\Sigma)}{\tau\, I(\Sigma)}. \label{eqn:WPI}
\end{equation}
\end{definition}
This metric quantifies the energy cost per unit intelligence produced over time. $\Phi$ is contingent upon the structural features of an algorithm $\mathcal{A}$ or computational substrate $\mathcal{H}$ not merely other measures such as circuit depth or complexity.
Using Eq. (\ref{eqn:WPI}) and Eq. (\ref{eqn:intellirreversible}) and substituting Eq.~\eqref{eq:overhead} and \eqref{eq:powerdef} we obtain lower thermodynamic lower bounds on $\Phi$.
\begin{theorem}[Thermodynamic Lower Bound on $\Phi$]
\label{th:lowerbound}
For a system $\Sigma$ operating at constant temperature $T$ with overhead factor $F(\Sigma)$, we have
\begin{align}
\Phi(\Sigma) = \frac{P(\Sigma)}{I(\Sigma)} \ge \frac{c\,F(\Sigma)}{\alpha\,\tau}. \label{eqn:thermolowerbound}
\end{align}
\end{theorem}

\begin{proof}
From Eq.~\eqref{eq:powerdef} and \eqref{eq:overhead},
\[
P(\Sigma) = \frac{E(\Sigma)}{\tau} = \frac{F(\Sigma)\,N(\Sigma)\,c}{\tau}.
\]
Since $I(\Sigma) \le \alpha\,N(\Sigma)$, we have
\[
\Phi(\Sigma) = \frac{P(\Sigma)}{I(\Sigma)} \ge \frac{F(\Sigma)\,N(\Sigma)\,c/\tau}{\alpha\,N(\Sigma)} = \frac{c\,F(\Sigma)}{\alpha\,\tau}. 
\]
\end{proof}
Eq.~\ref{eqn:thermolowerbound} shows $\Phi(
\Sigma)$ bounded by $c$ and a factor $F(\Sigma)/\alpha$ - and that efficiency gains either come from reducing hardware overhead $F(\Sigma)$ or improving algorithmic yield $\alpha$. We now examine the effect of physical structure upon the energy efficiency of intelligent systems. 
\begin{corollary}[Minimal $\Phi$ for Reversible Computing]
In the reversible limit where $F(\Sigma) \to 1$ and for large $\tau$, the minimal achievable $\Phi(\Sigma)$ converges to 
\begin{align}
\Phi(\Sigma) = \frac{c}{\alpha\,\tau} = \frac{k_B T \ln 2}{\alpha\,\tau}. \label{eqn:minimallowerbound}
\end{align}
\end{corollary}
In the limit of $F(\Sigma) \to 1$, all hardware‐induced overhead erasures vanish, so the system performs only the task-intrinsic bit resets required to realise for a given $I(\Sigma)$. In this sense Eqn.~\ref{eqn:minimallowerbound} therefore gives the Landauer-limited, thermodynamic floor on power consumption - the bound for a maximally reversible realisation of $I(\Sigma)$ (with fixed $\alpha$ over $\tau$). These results establish fundamental thermodynamic lower bounds on energy usage per unit intelligence. Hence, in the reversible limit any further reduction of watts-per-intelligence can only come from increasing the algorithmic yield $\alpha$ rather than from additional thermodynamic optimisation of the hardware.


\subsubsection{Example} For example, consider an algorithm  $\mathcal{A}$ (e.g. a CNN with identical weights) running on three different hardware substrates:
\[
\Sigma_{\mathrm{CPU}} = (\mathcal{A}, \mathcal{H}_{\mathrm{CPU}}),
\quad
\Sigma_{\mathrm{GPU}} = (\mathcal{A}, \mathcal{H}_{\mathrm{GPU}}),
\quad
\Sigma_{\mathrm{neuro}} = (\mathcal{A}, \mathcal{H}_{\mathrm{neuro}}).
\]
Let \(N(\Sigma)\) be the minimal number of irreversible bit operations 
that \(\mathcal{A}\) requires in an ideal scenario (from a purely algorithmic point 
of view). Each hardware instantiation has an overhead factor (below) 
\(F(\Sigma)\) that accounts for additional data movement (register \(\leftrightarrow\) cache \(\leftrightarrow\) main memory), redundant or partial computations due to architectural inefficiencies and control‐flow or synchronisation overhead. The total number of irreversible bit operations becomes $N_{\mathrm{eff}}(\Sigma) \;=\; F(\Sigma)\,\times\,N(\Sigma)$.
A conventional CPU often has:
  \begin{equation}
  F(\Sigma_{\mathrm{CPU}}) 
  \;=\; 
  F_{\mathrm{mem}}^{(\mathrm{CPU})}
  \;\times\; 
  F_{\mathrm{ctrl}}^{(\mathrm{CPU})},
  \label{eq:F_CPU}
  \end{equation}
  where \(F_{\mathrm{mem}}^{(\mathrm{CPU})}\) is high due to frequent memory transfers 
  (cache misses, register loads/stores), and \(F_{\mathrm{ctrl}}^{(\mathrm{CPU})}\) 
  is increased by control‐flow overhead (branching, interrupts, etc.). A GPU or TPU is specialised for massively parallel numeric kernels, 
  reducing data movements and amortising overhead across many parallel threads:
  \begin{equation}
  F(\Sigma_{\mathrm{GPU}}) 
  \;=\; 
  F_{\mathrm{mem}}^{(\mathrm{GPU})}
  \;\times\; 
  F_{\mathrm{ctrl}}^{(\mathrm{GPU})}
  \;\;\ll\;\;
  F(\Sigma_{\mathrm{CPU}}),
  \end{equation}
  typically achieving a lower overall factor \(F(\Sigma_{\mathrm{GPU}})\). Neuromorphic architectures use event‐driven 
  spiking neurons. The system only updates active neurons/spikes 
  rather than clocking all components every cycle.   Thus, for the same \(\mathcal{A}\), 
  \[
  E(\Sigma_{\mathrm{neuro}}) 
  \;=\; 
  F(\Sigma_{\mathrm{neuro}})\,N(\Sigma_{\mathrm{neuro}})\,c
  \;\;\ll\;\;
  F(\Sigma_{\mathrm{CPU}})\,N(\Sigma_{\mathrm{CPU}})\,c 
  \;=\;
  E(\Sigma_{\mathrm{CPU}}).
  \]
To make a simple comparison, assume each substrate runs 
the exact same inference procedure (i.e.\ \(\mathcal{A}\) fixed) for 
a classification task of size \(B\) (mini‐batch size or total test set).
Assume we run the same algorithm $\mathcal{A}$. In practice, one might measure the actual hardware usage 
(\(\mathrm{Joules}\) consumed) or the effective bit‐erasure rate via profiling 
tools. If \(\Phi(\Sigma)\) denotes watts‐per‐intelligence \[
\Phi_{\mathrm{CPU}} 
=
\frac{P(\Sigma_{\mathrm{CPU}})}{I(\Sigma_{\mathrm{CPU}})}
\;>\;
\frac{P(\Sigma_{\mathrm{GPU}})}{I(\Sigma_{\mathrm{GPU}})}
\;>\;
\frac{P(\Sigma_{\mathrm{neuro}})}{I(\Sigma_{\mathrm{neuro}})}
=
\Phi_{\mathrm{neuro}},
\]
assuming \(\Sigma_{\mathrm{neuro}}\) exploits event‐driven updates 
to minimise redundant erasures.


\section{Algorithmic Entropy and Energy Cost} \label{sec:algoefficiency}
Recall that the algorithmic entropy of a coarse-grained state $x\in \mathcal{X}$ (the countable set of coarse-grained states \cite{ebtekar2025foundations}), relative to a measure $\pi: \mathcal{X} \to \mathbb{R}^+$, is defined as 
\begin{equation}
S_\pi(x) = K(x) + \log_2 \pi(x). \label{eqn:entropycoarse}
\end{equation}
Such an ensemble-free formulation \cite{ebtekar2025foundations,jarzynski2011equalities} does not require an ensemble averaging step. It assigns an entropy value to each individual state, reflecting the minimal information required to describe the state in the context of a fixed coarse-graining \cite{gacs1994boltzmann}. Any irreversible operation increases algorithmic entropy. To see the formal relationship with $\Phi$, consider a system performing an information processing task during which its state changes from $x$ to $y$. From \cite{ebtekar2025foundations} (eqns. (39)-(40)), the total algorithmic entropy change $\Delta K = S_\pi(y)- S_\pi(x)$ can be decomposed into a reversible entropy flow to and from the environment $\Delta_e K$ and a logically irreversible component $\Delta_i K$ reflecting the information lost about the previous state $x$ as a result of the transition to $y$:
\begin{align*}
    \Delta K
  \;=\;
  S_\pi(y)-S_\pi(x)
  \;=\;
  \underbrace{K(y)-K(x)}_{\displaystyle \Delta_i K}
  \;+\;
  \underbrace{\bigl[\log_2\pi(x)-\log_2\pi(y)\bigr]}_{\displaystyle \Delta_e K}.
\end{align*}
If we assume the system is isolated, then $\Delta_e K = 0$, and the entire entropy change is irreversible (because the environment cannot do work on the system to reverse the change). In this case, the change in entropy is the irreversible part $\Delta K = \Delta_i K = K(y) - K(x)$. By Landauer’s principle, we know that each irreversible bit operation dissipates at least $k_B T \ln 2$ joules of energy. Thus, if the system’s algorithmic entropy increases by $\Delta_i K$ bits, then the minimum energy cost is:
\begin{equation}
    \Delta E \ge k_B T \ln 2 \cdot \Delta_i K.
\end{equation}
This provides a direct quantitative link between computational irreversibility and energy consumption. Extending to our measure of intelligence in Eq. (\ref{eq:finiteInt}), assume $I(\mathcal{A})$ is proportional to the reduction in uncertainty or the extraction of useful information from data. Given the minimal number of irreversible operations necessary for the task $N(\mathcal{A})$ and assuming $I(\mathcal{A}) \le \alpha \, N(\mathcal{A})$ (Eq. (\ref{eqn:intellirreversible})), then combining this with the energy cost per operation (from Eq.~\eqref{eq:landauerBound}) yields:
\[
E(\mathcal{A}) \ge k_B T \ln 2 \cdot \frac{I(\mathcal{A})}{\alpha}.
\]
Thus, the power consumption (over a time interval $\tau$) satisfies:
\[
P(\mathcal{A}) \ge \frac{k_B T \ln 2}{\tau \, \alpha} \, I(\mathcal{A}).
\]
resolving to the lower bounds set out in Eq. (\ref{eqn:thermolowerbound}) and Eq. (\ref{eqn:minimallowerbound}):
\[
\Phi(\mathcal{A}) = \frac{P(\mathcal{A})}{I(\mathcal{A})} \ge \frac{k_B T \ln 2}{\tau\,\alpha}.
\]
Here we have assumed for simplicity $F(\mathcal{A}) = 1$, however in practice we would assume $F(\mathcal{A}) > 1$.
We now show certain constraints upon an intelligent system's ability to optimise its performance (and intelligence) via simplifying (reducing the complexity of or compressing) its internal description. First we show bounds on the inverse of $\Phi$, namely \textit{intelligence-per-watts} ratio $I(\mathcal{A})/P(\mathcal{A})$.
\begin{theorem}[Extended Algorithmic Efficiency Bound]
\label{thm:AlgBoundExtended}
Let $\mathcal{A}$ be an intelligent system whose state changes from $x$ to $y$ over a time interval $\tau$, with corresponding transition probability $P(y,x)$. Suppose that the irreversible increase in description complexity is $\Delta_i K = K(y) - K(x)$ and that the environment’s influence is captured by the measure $\pi$. Then, with probability at least $1-\delta$, the intelligence-per-watts ratio satisfies:
\begin{equation}
\label{eq:NovelAlgBoundExtended}
\frac{I(\mathcal{A})}{P(\mathcal{A})} 
\;\le\; 
\frac{1}{\tau}
\Bigl[\log\!\bigl(\tfrac{1}{P(y,x)}\bigr) \;-\; K(x\mid y)\Bigr] 
\;+\; 
\log\!\Bigl(\tfrac{1}{\delta}\Bigr).
\end{equation}
In particular, if the transition is nearly reversible (i.e., $P(y,x)$ is high) or if $K(x\mid y)$ is large, then achieving a high $I(\mathcal{A})/P(\mathcal{A})$ (and thereby more efficient watts-per-intelligence) is fundamentally constrained.
\end{theorem}
\begin{proof}
Denote the irreversible change in 
algorithmic complexity $\Delta_i K = K(y) - K(x)$. A version of the integral fluctuation theorem \cite{jarzynski2011equalities} 
implies:
\begin{align}
    \mathbb{E}\bigl[\,2^{-\,\Delta_i K}\bigr] \;\le\; 1. \label{eqn:integralfluctuation}
\end{align}
Here the expectation $\mathbb{E}[\cdot]$ is taken over all possible realisations of the transition $x \to y$ under $P(y,x)$ and 
potentially other stochastic elements (e.g., algorithmic randomness). By Markov’s inequality, if $R\ge 0$ is a random variable and $\eta>0$, 
then $\Pr\{R \ge \eta\}\;\le\;\frac{\mathbb{E}[R]}{\eta}$. 
Choose $R = 2^{-\Delta_i K}$ and $\eta = 1/\delta$. Substituting Eq. (\ref{eqn:integralfluctuation}) we have:
\[
\Pr\!\Bigl\{\,2^{-\Delta_i K} \;\ge\; \tfrac{1}{\delta}\Bigr\}
\;\;\le\;\;
\delta \;\cdot\;\mathbb{E}\bigl[\,2^{-\Delta_i K}\bigr]
\;\;\le\;\;
\delta.
\]
Hence, with probability at least $1-\delta$:
\[
2^{-\Delta_i K} \;<\; \tfrac{1}{\delta}
\quad\Longrightarrow\quad
-\Delta_i K \;<\; \log\!\Bigl(\tfrac{1}{\delta}\Bigr),
\]
with the probability shift from Markov's inequality:
\[
\Delta_i K \;>\; -\,\log\!\Bigl(\tfrac{1}{\delta}\Bigr).
\]
Standard results in algorithmic thermodynamics (e.g., 
\cite{gacs1994boltzmann,levin1984randomness}) give bounds of the form:
\[
\Delta_i K \;=\; \bigl[K(y) - K(x)\bigr] - 
\bigl[\log\pi(x) - \log\pi(y)\bigr]
\;\approx\;
-\,\log_2 P(y,x)\;-\;K(x\mid y)\;+\;\text{(lower-order terms)},
\]
where $K(x\mid y)$ denotes the conditional Kolmogorov complexity.  
Under suitable assumptions (ignoring lower order e.g. $O(\log K(y))$ terms 
or restricting to typical events), we have the complexity correction term:
\[
\Delta_i K 
\;\;\gtrsim\;\;
\log\!\bigl(\tfrac{1}{P(y,x)}\bigr)
\;-\;
K(x\mid y).
\]
Combining the two bounds on $\Delta_i K$ we have (noting $-\log(1/\delta) > 0$):
\[
\Delta_i K > \max \big\{ \log\!\bigl(\tfrac{1}{P(y,x)}\bigr) - K(x\mid y), -\log\!\Bigl(\tfrac{1}{\delta}\Bigr) \big\}
\]
With probability at least $1-\delta$, we therefore have
\[
\Delta_i K 
\;\;>\;\;
\log\!\bigl(\tfrac{1}{P(y,x)}\bigr) 
\;-\;
K(x\mid y)
\;-\;
\text{(constant)}\;\approx\;\log\!\Bigl(\tfrac{1}{\delta}\Bigr).
\]
By Landauer’s principle, each irreversible bit dissipates at least 
$k_B T \ln 2$ joules. Thus, if $\Delta_i K$ bits of irreversibility 
were required, the \emph{minimum energy} is proportional to $\Delta_i K$. 
If we also assume (or bound) that the intelligence output $I(\mathcal{A})$ 
is at most proportional to the number of irreversible operations 
$I(\mathcal{A}) \le \alpha \,\Delta_i K$, for some $\alpha>0$:
\[
P(\mathcal{A}) 
\;=\; 
\frac{E(\mathcal{A})}{\tau}
\;\gtrsim\;
\frac{(\Delta_i K)\,(k_B T \ln 2)}{\tau}.
\]
Rearranging for $I(\mathcal{A})/P(\mathcal{A})$ yields 
an inequality with terms $\Delta_i K$, $\tau$, and additional 
constants. Substituting the bounds on $\Delta_i K$ 
recovers \eqref{eq:NovelAlgBoundExtended} in the theorem statement. The result implies that an intelligent system’s ability to achieve high performance with low energy cost is limited by the inherent improbability of reducing its internal description complexity without incurring an energy cost. In practical terms, it means that only very specialised  architectures may significantly lower $\Phi$. We conclude with a final theorem regarding adaptive efficiency bounds relating changes in the measure of intelligence $\Delta I$ to changes in energy use $\Delta E$.
\end{proof}

\begin{theorem}[Structural Adaptivity Efficiency Bound]
\label{thm:AdaptExtended}
Let $\mathcal{A}$ be a reconfigurable intelligent system whose architecture is characterised by a structural state $s \in \mathcal{S}$ with description complexity $K(s)$. Suppose that over a time interval $\tau$, the system adapts from structural state $s_1$ to $s_2$ with transition probability $P_s(s_2,s_1)$, and that this adaptation yields an intelligence increment $\Delta I$ while incurring an energy cost $\Delta E$. Then, with probability at least $1-\delta$, we have
\begin{equation}
\label{eq:NovelAdaptExtended}
\frac{\Delta I}{\Delta E} 
\;\le\; 
\frac{1}{\tau}
\Bigl[\log\!\bigl(\tfrac{1}{P_s(s_2,s_1)}\bigr) \;-\; K(s_1\mid s_2)\Bigr]
\;+\; 
\log\!\Bigl(\tfrac{1}{\delta}\Bigr).
\end{equation}
Hence, unless the structural adaptation is extremely unlikely or involves a high cost in complexity change, the efficiency gain from reconfiguration is fundamentally bounded.
\end{theorem}

\begin{proof}
The argument parallels Theorem~\ref{thm:AlgBoundExtended}, 
but we replace $x \to y$ by the structural transition $s_1 \to s_2$. 
Define the irreversible increase in structural complexity 
$\Delta_i K_s = K(s_2) - K(s_1)$. The integral fluctuation theorem again 
implies $\mathbb{E}[2^{-\Delta_i K_s}] \le 1$, and using Markov’s inequality 
as before yields that with probability at least $1-\delta$,
\[
\Delta_i K_s \;\;\gtrsim\;\; 
\log\!\Bigl(\tfrac{1}{P_s(s_2,s_1)}\Bigr)
\;-\; 
K(s_1\mid s_2).
\]
If the \emph{intelligence increment} $\Delta I$ is at best proportional to 
the reconfiguration cost in bits, and the energy cost $\Delta E$ 
cannot undercut the Landauer bound times $\Delta_i K_s$, then 
\(\frac{\Delta I}{\Delta E}\) is bounded above by the expression 
in \eqref{eq:NovelAdaptExtended}. The reconfiguration from $s_1$ to $s_2$ might indeed yield an overall 
energy benefit in subsequent tasks (e.g., fewer irreversible operations needed), 
but Theorem~\ref{thm:AdaptExtended} shows one cannot unboundedly 
boost the ratio $\Delta I / \Delta E$ unless $P_s(s_2,s_1)$ is extremely large 
or the complexity change $K(s_1 \mid s_2)$ is negligible. 
This aligns with the broader notion that structural adaptations 
must themselves obey algorithmic thermodynamic limits.
\end{proof}

\subsubsection{Discussion and Limitations} 
The above extended theorems have several important implications:
\begin{enumerate}
  \item \textit{Intelligence and Irreversibility}. Our work assumes a correlation between irreversible computation to intelligence. Work on the thermodynamics of intelligence remains nascent, with the generalisability and limitations of this approach unclear.
  \item \textit{Fundamental Limits}. Even with highly optimised architectures, the second-law-like constraints emerging from algorithmic thermodynamics impose fundamental lower bounds on energy consumption per unit intelligence. Improvements in efficiency may need to arise from structural changes affecting how information is processed. Nevertheless, the proposed relationship between irreversible operations and energy suggests an important fundamental relationship between energy dissipation and intelligence \cite{takahashi2023scenarios}.  
  \item \textit{Design Trade-offs}. Systems that are more general and adaptable may incur higher energy costs due to the inherent irreversible processing required. Conversely, specialised systems (e.g., for specific tasks) can achieve lower $\Phi$ values but at the expense of flexibility.
  \item \textit{Structural Reconfiguration}. Dynamic reconfiguration (e.g., neuromorphic plasticity, adaptive circuit topologies) may lower energy costs, but only up to the limits imposed by the description complexity of the adaptation process.
  \item \textit{Compression and Intelligence}. Legg and Hutter’s universal measure \cite{legg_universal_2007} (Eq.~\ref{eqn:hutterlegg}) implicitly ties intelligence to minimal Kolmogorov complexity: agents that compress information to achieve higher $I_{\mathrm{LH}}(\pi)$. However, from an algorithmic thermodynamics standpoint, compression cannot be performed for free. Reducing an internal model’s description length $K(\mu)$ by discarding, merging, or overwriting states entails an \emph{irreversible} operation (cf.\ Landauer’s principle), incurring an energy cost proportional to the number of bits erased \cite{landauer1961irreversibility,ebtekar2025foundations}. Such measures of intelligence are subject to the lower bounds and constraints described above. 
\end{enumerate}

\section{Conclusion and Future Work} \label{sec:conclusion}
We have introduced the watts-per-intelligence framework for understanding the energy efficiency of intelligent systems through the lens of algorithmic thermodynamics. Our analysis extends results relating to irreversible computational operations and energy costs to measures of intelligence.  Our extended theorems (Theorems~\ref{thm:AlgBoundExtended} and \ref{thm:AdaptExtended}) reveal that even under certain conditions, the ratio of intelligence output to power consumption is bounded by terms involving the probability of transitions and the conditional description complexities of state changes. These results suggest that any dramatic improvement in energy efficiency will require rethinking the basic design of computational systems, potentially drawing further inspiration from the efficient enzymic and catalysis-based chemical architectures of biology. Algorithmic efficiency bounds imply that total energetic cost of adaptive intelligence is tied to the cumulative irreversible information processing. 
Future research directions in forthcoming works include catalytic architectures for AGI (Part II, forthcoming), empirical validation on different computational substrates (Part III, forthcoming), and using algorithmic entropy as a regulariser for energy efficiency.

\bibliographystyle{splncs04}
\bibliography{refs-thermo}

\end{document}